\documentclass[
]{ceurart}

\usepackage{listings}
\usepackage{graphbox}
\DeclareRobustCommand{\DLLogo}{%
  \begingroup\normalfont
  \kern-1.75pt\includegraphics[align=c,height=1.25\baselineskip]{dl}\kern-1.5pt%
  \endgroup
}

\usepackage{mathtools}

\usepackage{amsthm, thmtools}
\theoremstyle{plain}
\newtheorem{thm}{Theorem}
\newtheorem{coro}{Corollary}
\newtheorem{lem}{Lemma}
\theoremstyle{definition}
\newtheorem{defn}{Definition}
\theoremstyle{plain}
\newtheorem{example}{Example}

\usepackage[printonlyused]{acronym}

\usepackage[noabbrev, capitalise]{cleveref}
\crefname{ax}{Axiom}{Axioms}
\crefname{thm}{Theorem}{Theorems}
\crefname{defn}{Definition}{Definitions}
\crefname{coro}{Corollary}{Corollaries}
\crefname{prop}{Proposition}{Propositions}
\crefname{lem}{Lemma}{Lemmata}
\crefname{example}{Example}{Examples}

\usepackage{macros}

\usepackage[arxiv]{tagging}

\tagged{arxiv}{
\RenewDocumentCommand \logonote { m } {}
\conference{}%
}

\begin{document}

\copyrightyear{2023}
\copyrightclause{Copyright for this paper by its authors.
	Use permitted under Creative Commons License Attribution 4.0
	International (CC BY 4.0).}

\tagged{dl}{
\conference{\DLLogo{} DL 2023: 36th International Workshop on Description Logics,
	September 2--4, 2023, Rhodes, Greece}
} 
\tagged{arxiv}{
}

\title{Description Logics Go Second-Order — Extending EL with Universally Quantified Concepts}

\author[1]{Joshua Hirschbrunn}[%
email=joshua.hirschbrunn@uni-ulm.de
]
\address[1]{University of Ulm,
   James-Franck-Ring, 89081 Ulm, Germany}
\author[1]{Yevgeny Kazakov}[%
email=yevgeny.kazakov@uni-ulm.de
]

\begin{abstract}
    The study of Description Logics have been historically mostly focused on features that can be translated to decidable fragments of first-order logic.    
    In this paper, we leave this restriction behind and look for useful and decidable extensions outside first-order logic. We introduce universally quantified concepts, which take the form of variables that can be replaced with arbitrary concepts, and define two semantics of this extension.
    A schema semantics allows replacements of concept variables only by concepts from a particular language, giving us axiom schemata similar to modal logics.    
    A second-order semantics allows replacement of concept variables with arbitrary subsets of the domain, which is similar to quantified predicates in second-order logic.

    To study the proposed semantics, we focus on the extension of the description logic $\EL$.
    We show that for a useful fragment of the extension, the conclusions entailed by the different semantics coincide, allowing us to use classical $\EL$ reasoning algorithms even for the second-order semantics.
    For a slightly smaller, but still useful, fragment, we were also able to show polynomial decidability of the extension.    
    This fragment, in particular, can express a generalized form of role chain axioms, positive self restrictions, and some forms of (local) role-value-maps from KL-ONE, without requiring any additional constructors.
\end{abstract}


\maketitle


\section{Introduction and Motivation}
\label{sec:introduction}

After the famous undecidability result of the early knowledge representation language used in the KL-ONE system \cite{shmidt-schauss89klone-undecidable}, it has been commonly agreed that \acp{DL} should be restricted to decidable fragments of \ac{FOL},  preferably of lower complexity.
The subsequent search for suitable languages has resulted in a large spectrum of \acp{DL} ranging from simple tractable languages, such as $\EL$ \cite{baader2005pushing} and $\DLLite$ \cite{dl-lite} to very expressive languages, such as $\SROIQ$ \cite{HKS06-sroiq}, used as the basis of the Web ontology language OWL 2 \cite{owl2-semantics}. 
After these extensive studies, finding new interesting fragments of \ac{FOL}, has become increasingly difficult.
That is why, in this paper, we look at some features from \ac{SOL}, specifically at \emph{universally quantified concepts}.
Syntactically, we introduce new
\emph{concept variables}, which can be used in place of classical concepts. 
For example, an axiom $\exists{\textsf{owns}}.(X\sqcap\textsf{Pet})\sqsubseteq\exists{\textsf{feeds}}.X$ expresses that everyone who owns a pet from a set of objects $X$ must feed someone from this set.
Such concept variables $X$ are assumed to be implicitly \emph{universally quantified}, that is, the above property should hold for \emph{every reasonable choice} of a set $X$.

One possible way of defining semantics of our language extension, is to regard axioms with concept variables as \emph{axiom schemata}.
Under this reading, the reasonable choices of $X$ in the above axiom, are sets defined by \emph{concepts} of a specific form, for example, all (atomic) concepts from an ontology, or all concepts that can be built in a particular \ac{DL}, such as $\EL$.
After replacing each variable with every possible concept, one obtains a (possibly infinite) set of ordinary \ac{DL} axioms, for which the logical entailment is defined in the standard way.

Axiom schemata is certainly not a new concept.
In fact, most logic languages, including \ac{PL}, \ac{FOL}, and \acp{ML} were originally defined \emph{axiomatically}.
For example, \ac{PL} can be defined using 3 axiom schemata: $X\to X$, $X\to(Y\to X)$, and $(X\to(Y\to Z))\to ((X\to Y)\to (X\to Z))$, in which every variable can be replaced with any propositional formula.
Many \acp{ML} were defined by extending \ac{PL} with new axiom schemata, for example, K4: $\Box X\to\Box\Box X$ (see, e.g., \cite{blackburn2001modal}).
In the context of \acp{DL}, axiom schemata occur in different contexts: Nominal Schemas \cite{carral2023efficient,krotzsch2011better,carral2013towards} introduce (nominal) variables that can be replaced with individual names.
\acp{ODP} \cite{xiang2015ontorat,skjaeveland2018practical,kindermann2018generating,krieg2019generic,kindermann2019comparing,skjaeveland2019pattern,borgida2012experiences},
use concept variables to define axiom templates, which can be used to generate ordinary axioms for specific applications.
Axiom schemata with existentially-quantified concepts can be also used to define operations, such as interpolants or most-common subsumer \cite{colucci2010secondorder}.
To obtain decidability results, most works restrict the set of values that can be used for replacement of variables so that the resulting set of (ordinary) axioms is \emph{finite}.
However, even for unrestricted axiom schemata, such as those found in \acp{PL} and \acp{ML}, it can be often shown that it is sufficient to replace variables with only finitely-many formulas found in (or built using) the entailment to be proved.
For \acp{ML}, this \emph{sub-formula property} usually follows from \emph{cut-free sequent-style} calculi (see, e.g., \cite{fitting07proof-systems}).

The main advantage of the schema semantics for concept variables, is that any \ac{DL} reasoning procedure could, in principle, be used for checking entailment in the corresponding \ac{DL} extension, by systematically generating instances of axiom schemata and checking entailment from the resulting increasing sets of ordinary axioms.
This immediately implies \emph{semi-decidability} of the schema extension and, in the cases when a form of sub-formula property can be proved, also \emph{decidability} and \emph{complexity} results.
Schema semantics, however, also has disadvantages, one of which, is that the entailments depend on the choices of concepts used for variable replacements.
Consider, for example, axiom schema $\top\sqsubseteq\exists r.X$. 
If we only allow $\EL$ concepts $C$ for variable replacements, the resulting set of axioms is satisfiable since there is a model (with one element) that interprets every $\EL$ concept (including every $\exists r.C$) by the whole domain.
If, however, we additionally allow $X$ to be replaced by $\bot$, i.e., by taking concepts from $\EL^\bot$, then, clearly, this schema becomes unsatisfiable.
As a consequence, the $\EL$ axiom schema obtains new $\EL$ logical conclusions (e.g., $A\sqsubseteq B$) when viewing it in a context of a larger language (e.g., $\EL^\bot$).
In the next section we give a similar example when an $\EL$ schema obtains new $\EL$ consequences with $\ALC$ concept replacements, this time, without making axioms inconsistent.\footnote{To see why this could be a concern, imagine that the axiom schemata of \ac{PL} would obtain new logical consequences in the language of \ac{PL} when allowing variables to be replaced by modal formulas!}

Motivated by the above consideration, we define another, \emph{language-independent} semantics of concept variables.
Specifically, an axiom with concept variables is satisfied by an interpretation if for \emph{every} assignment of concept variables to \emph{subsets} of its domain, the axiom is satisfied in the extension of this interpretation in which the variables are interpreted by these subsets (as ordinary concept names).
For example, axiom $\top\sqsubseteq\exists r.X$ is \emph{not} satisfied in any interpretation $\I$ since for the assignment of $X\mapsto\emptyset$, this axiom does not hold if we extend $\I$ by $X^\I=\emptyset$.
On the other hand, axiom $X\sqsubseteq\exists r.X$ is satisfied in exactly those interpretations $\I$ such that $r^\I$ is a reflexive relation.
Similarly, axiom $\exists{\textsf{owns}}.(X\sqcap\textsf{Pet})\sqsubseteq\exists{\textsf{feeds}}.X$ is satisfied in exactly those interpretations $\I$ such that $\tuple{x,y}\in\textsf{owns}^\I$ and $y\in\textsf{Pet}^\I$ imply ${\tuple{x,y}}\in\textsf{feeds}^\I$.\footnote{Note that this property can be expressed in $\SROIQ$ using a fresh role $\textsf{pet}$: $\textsf{Pet} \sqsubseteq \exists \textsf{pet}.\textsf{Self}$ and $\textsf{pet} \circ \textsf{owns} \sqsubseteq \textsf{feeds}$}
Essentially, under this semantics, concept variables correspond to universally-quantified second-order (unary) predicates.
Hence, we refer to this semantics as \emph{second-order semantics}.

It is easy to see that the second-order semantics is \emph{stronger} than the schema semantics.
Indeed, if an axiom is satisfied in an interpretation under the second-order semantics, then for every replacement of variables by concepts, the resulting axioms are also (classically) satisfied in this interpretation because variables can be assigned to interpretations of these concepts in the second-order semantics.
The converse is not true, as has been shown on the example of axiom $\top\sqsubseteq\exists r.X$ for the $\EL$-schema semantics.
Furthermore, whereas schema entailment can be always reduced to first-order entailment (from a possibly infinite set of formulas), there are axioms whose models under second-order semantics cannot be expressed by first-order formulas.
For example, the well-known \ac{ML} \emph{McKinsey axiom} $\Box \Diamond X \rightarrow \Diamond \Box X$, which can be written in the \ac{DL} $\mathcal{ALC}$ as $\forall r.\exists r.X \sqsubseteq \exists r.\forall r.X$, 
cannot be translated into even an infinite set of first-order formulas that holds in exactly the same interpretations (of $r$) \cite{goldblatt1975firstorder}.
Hence, even semi-decidability for the second-order extension of $\mathcal{ALC}$ seems to be an open question.

To combine advantages of the two proposed semantics (semi-decidability and schema language independence), it makes sense to look for restricted use of concept variables in \acp{DL} for which the entailment under two semantics \emph{coincide}.
In this paper, we describe a useful extension of the \ac{DL} $\EL$ with concept variables for which this is the case (\cref{sec:coincide}), and a further restriction, for which these (equivalent) entailments are \emph{polynomially decidable} (\cref{sec:decidabilty}).
Even though the use of concept variables in this fragment is restricted, it can express several other \ac{DL} features, such as role chain axioms, positive self restrictions, reflexive roles, some forms of (local) role-value-maps from KL-ONE, and their generalizations (see \cref{sec:conclusions}).

The main idea for proving these results, is to show that the standard $\EL$ canonical model \cite{baader2005pushing} for axioms obtained by replacing concept variables with certain \emph{relevant} $\EL$ concepts, is also a model of the original axioms under the second-order semantics.
This property holds because for our restricted form of axioms, it is sufficient to replace concept variables with only \emph{singleton subsets} of the domain to capture the second-order semantics.\footnote{This means, in particular, that our restricted axioms can be translated to \ac{FOL}. Note that any fragment for which the second-order entailment coincides with the schema entailment can be translated to \ac{FOL} using the standard \ac{FOL} translation for the schema instances.}
For example, if $X\sqsubseteq\exists r.X$ holds for all replacements of $X$ with singleton subsets of the domain (that is, when $r$ is reflexive) then this axiom also holds for replacements of $X$ with any larger set.
To obtain the decidability and complexity result, we define a further restriction for which it is sufficient to use only polynomially-many relevant concepts for checking the schema entailment.

\begin{taggedblock}{dl}
    Because of limited space, we omit most proofs in this paper. They can be found in the accompanying technical report \cite{hirshbrunn2023elx-tr}. 
\end{taggedblock}


\section{Schema Semantics and Second-Order Semantics}
\label{sec:semantics}

We start by formally defining our extension of \ac{DL}s with concept variables:

\begin{defn}[Syntax]
    The \emph{syntax} of \ac{DL}s with concept variables consists of disjoint and countably infinite sets $N_C$ of \emph{concept names}, $N_R$ of \emph{role names}, and $N_X$ of \emph{concept variables}. 
    Given a base $\ac{DL}$ $\L$ that is a fragment of $\SROIQ$, such as $\EL$ and $\ALC$, we define by $\LX$ its corresponding \emph{extension with concept variables}, in which the elements from $N_X$ can be used as concepts.
    We define, an \emph{$\LX$-ontology} as a (possibly infinite) set $\KB$ of $\LX$-axioms.

    Let $\ex$ be either an $\LX$-concept, an $\LX$-axiom, or an $\LX$-ontology.
    We denote by $\sub(\ex)$ (all) subconcepts\footnote{\emph{Subconcepts} are defined to be substrings of the expression that are valid concepts.} of $\ex$.
    For $\LX$-concepts and $\LX$-axioms, we split $\sub(\ex)$ into $\sub^+(\ex)$ and $\sub^-(\ex)$ the set of concepts that \emph{occur positively}, respectively \emph{negatively} in $\ex$\footnote{A concept occurs positively (negatively) in an axiom, if it occurs on the right side of the axiom under even (odd) number of nested negations or on the left side under odd (even) number of nested negations.} (i.e.\ $\sub(\ex)=\sub^+(\ex)\cup\sub^-(\ex)$).
    We denote by $\vars(\ex)=\sub(\ex)\cap N_X$ the set of \emph{concept variables occurring in $\ex$}.
    We say that $\ex$ is \emph{ground} if $\vars(\ex)=\emptyset$.
    
    A \emph{(concept variable) substitution} is a mapping $\theta = [X_1/C_1, \dots, X_n/C_n]$ that assigns concepts $C_i$ to concept variables $X_i$ ($1\le i\le n$). 
    We say that $\theta$ is \emph{ground} if all $C_i$ are ground $(1\le i\le n)$.
    We denote by $\theta(\ex)$ the result of \emph{applying the substitution} to $\ex$, defined in the usual way.
\end{defn}

From now on we assume that \ac{DL} concepts, axioms, and ontologies may contain concept variables.
When this is not the case, we explicitly say that these concepts, axioms, or ontologies are \emph{ground} or \emph{classical}.

\begin{example}
    \label{example:syntax}
    Consider the following two (non-ground) axioms:
    \begin{align}
		\alpha=\exists r.X \sqcap \exists s.X & \sqsubseteq \exists t.X   \label{ax:EL},  \\
        \beta=\exists r.X \sqcap \exists s.Y & \sqsubseteq \exists t.(X\sqcup Y). \label{ax:ALC}
    \end{align}
  Axiom $\alpha$ belongs to $\ELX$ whereas axiom $\beta$ belongs to $\ALCX$ due to the use of concept disjunction $\sqcup$.
  Further, $\sub^+(\alpha) =\set{\exists t.X, X}$,
  $\sub^-(\beta) =\set{\exists r.X\sqcap\exists s.Y, \exists r.X, \exists s.Y, X, Y}$,
  and $\vars(\beta)=\set{X,Y}$.
  Finally, for a (non-ground) substitution $\theta=[X/X\sqcup Y]$, we have:
  \begin{align}
  \theta(\alpha) = \exists r.(X\sqcup Y)\sqcup\exists s(X\sqcup Y)\sqsubseteq\exists t.(X\sqcup Y).
  \end{align}
\end{example}

Intuitively, axioms $\alpha$ and $\beta$ can be thought as \emph{axiom schemata} representing all axioms obtained by replacing the concept variables $X$ and $Y$ with ordinary (ground) concepts.
However, the choice of such ground concepts is not always obvious. 
Should variables be replaced by just atomic concepts? 
Or only by concepts appearing in the given ontology?
Or by any concepts that can be constructed in a particular \ac{DL}?
Clearly, each of these choices may result in different logical consequences and algorithmic properties of the resulting schema languages.
To handle all such choices, we provide a general (parameterized) definition of schema semantics.

\begin{defn}[Schema Semantics]
    \label{def:semantics:schema}
    Let $\KB$ be an $\LX$-ontology for some \ac{DL} $\L$, and $H$ a (possibly infinite) set of $\L$-concepts called a \emph{concept base}.
    For an $\LX$-axiom $\alpha$, and $\LX$-ontology $\KB$,
    define by $\alpha_{\downarrow H}=\set{\alpha[X_1/C_1, \dots
        X_n/C_n]\mid X_i\in\vars(\alpha)\: \& \: C_i\in H}$ 
        and $\KB_{\downarrow H}=\bigcup_{\alpha\in\KB} \alpha_{\downarrow H}$
        the set of \emph{$H$-ground instances} of $\alpha$ and $\KB$, respectively.

    We write $\KB\models^{\ast}_H\alpha$ if $\KB_{\downarrow H}\models\alpha_{\downarrow H}$ and say that $\alpha$ is a logical consequence of $\KB$ under the \emph{schema semantics} (for a concept base $H$).
    Finally, we write $\alpha_{\downarrow\L}$, $\KB_{\downarrow\L}$ and $\KB\models^{\ast}_\L\alpha$ instead of $\alpha_{\downarrow H}$, $\KB_{\downarrow H}$ and $\KB\models^{\ast}_H\alpha$, respectively, if $H$ is the set of \emph{all} $\L$-concepts.
\end{defn}

Clearly, if the concept base $H$ is \emph{finite}, entailment under the schema semantics for $H$ can be reduced to the standard \ac{DL} entailment.

\begin{example}[Example~\ref{example:syntax} Continued]
\label{ex:schema:conservativity}
Notice that
$\set{\beta}\models^{\ast}_H\alpha$ for any concept base $H$.
Indeed, take any $\alpha'=\theta(\alpha)=\exists r.C\sqcap\exists s.C\sqsubseteq\exists t.C\in\alpha_{\downarrow H}$.
Then $\beta_{\downarrow H}\owns\beta[X/C,Y/C]=\exists r.C\sqcap\exists s.C\sqsubseteq\exists t.(C\sqcup C)\models
\alpha'$.
Also, notice that if $H$ is closed under concept disjunctions (i.e., $C\in H$ and $D\in H$ imply $C\sqcup D\in H$) then
$\set{\alpha}\models^{\ast}_H\beta$.
Indeed, take any $\beta'=\theta(\beta)=\exists r.C\sqcap\exists s.D\sqsubseteq\exists t.(C\sqcup D)\in\beta_{\downarrow H}$ for $\theta=[X/C,Y/D]$. 
Since $H$ is closed under disjunction, we have $\alpha_{\downarrow H}\owns\alpha[X/C\sqcup D]=\exists r.(C\sqcup D)\sqcap\exists s.(C\sqcup D)\sqsubseteq\exists t.(C\sqcup D)\models\beta'$.
Consequently, $\set{\alpha}\models^{\ast}_\ALC\beta$.
However, it can be shown that $\set{\alpha}\not\models^{\ast}_\EL\beta$.
To prove this, consider the  $\EL$ ontology:
\begin{equation}
\KB = \set{~A\sqsubseteq\exists r.B\sqcap\exists s.C,\quad
\exists t.B\sqsubseteq D,\quad\exists t.C\sqsubseteq D~}
\label{eq:conservativity:KB}
\end{equation}
It is easy to see that $\KB\cup\beta_{\downarrow\EL}\supseteq \KB\cup\set{\beta'}\models A\sqsubseteq D$ for $\beta'=\beta[X/B,Y/C]=\exists r.B\sqcap\exists s.C\sqsubseteq \exists t(B\sqcup C)$.
Hence $\KB\cup\set{\beta}\models^{\ast}_{\EL} A\sqsubseteq D$.
We show that $\KB\cup\set{\alpha}\not\models^{\ast}_{\EL} A\sqsubseteq D$, which, in particular, implies that $\set{\alpha}\not\models^{\ast}_\EL\beta$, for otherwise $\KB\cup\set{\alpha}\models^{\ast}_{\EL}\KB\cup\set{\beta}\models^{\ast}_{\EL} A\sqsubseteq D$.

To prove that $\KB\cup\set{\alpha}\not\models^{\ast}_{\EL} A\sqsubseteq D$, consider the interpretation $\I=(\Delta^\I,\cdot^\I)$ defined by:
$\Delta^{\I}  = \set{a,b,c}$, 
$A^\I = \set{a}$,
$B^\I = \set{b}$,
$C^\I = \set{c}$,
$r^\I = \set{\tuple{a,b}}$,
$s^\I = \set{\tuple{a,c}}$,
$t^\I = \set{\tuple{a,a}}$,
and $E^\I=h^\I=\emptyset$ for all remaining $E\in N_C$ and $h\in N_R$. 
Clearly, $\I\models\KB$ and $\I\not\models A\sqsubseteq D$.
It remains to prove that $\I\models\alpha_{\downarrow\EL}$.
Take any $\alpha[X/F]=\exists r.F\sqcap\exists s.F\in\exists t.F \in\alpha_{\downarrow\EL}$ and any $d\in (\exists r.F\sqcap\exists s.F)^\I$. 
Then $d=a$ and $\set{b,c}\subseteq F^\I$ by definition of $r^\I$ and $s^\I$. 
But then $F=\top$ since $\sizeof{F^\I}\le 1$ for all other $\EL$ concepts $F$.
Then $d=a\in (\exists t.\top)^\I=(\exists t.F)^\I$ as required.

Combining the above observations, we obtain:
$\KB\cup\set{\alpha}\models^{\ast}_{\ALC}\KB\cup\set{\beta}\models^{\ast}_{\EL}\KB\cup\set{\beta}\models^{\ast}_{\EL} A\sqsubseteq D$, however, $\KB\cup\set{\alpha}\not\models^{\ast}_{\EL} A\sqsubseteq D$.
\end{example}

Example~\ref{ex:schema:conservativity} shows that, for schema semantics, an ontology formulated in one \ac{DL} may have different conclusions, even in the same \ac{DL}, when the concept base is extended to a larger language.
This goes against the usual understanding of logic, as this means that the consequences of an ontology are not determined by the ontology alone. 
To mitigate this problem, we consider the second-order semantics that is independent of a concept base.

\begin{defn}[Second-Order Semantics]
\label{def:semantics:second-order}
    Let $\I = (\Delta^{\I}, \cdot^{\I})$ be an
	interpretation. 
    A \emph{valuation}\index{valuation} for $\I$ (also called
    a variable assignment) is a mapping
	$\eta$ that assigns to every variable $X\in N_X$ a subset $\eta(X) \subseteq \Delta^{\I}$. 
    The \emph{interpretation of concepts} $C^{\I, \eta}$ and \emph{satisfaction of axioms} $\I \models^{2}_{\eta} \alpha$ \emph{under $\I$ and $\eta$} is defined in the same way as for the standard \ac{DL} semantics by treating concept variables $X$ as ordinary concept names interpreted by $\eta(X)$.    
    We write $\I \models^{2} \alpha$ if $\I\models^{2}_\eta\alpha$ for every \emph{valuation} $\eta$.
    Finally, for an ontology $\KB$, we write $\I\models^{2}\KB$ if $\I\models^{2}\beta$ for every $\beta\in\KB$, and we write $\KB\models^{2} \alpha$ if $\I\models^{2}\KB$ implies $\I\models^{2}\alpha$.
\end{defn}

Alternatively to the model-theoretic definition (Definition~\ref{def:semantics:second-order}) it is possible to define the second-order semantics by a translation to \ac{SOL}.
This translation is simply the normal translation to \ac{FOL}, treating concept variables as second-order unary predicates and then universally quantifying over these predicates. 
For example, $C \sqsubseteq \exists r.X \sqcap Y$ translates to $\forall X.\forall Y.[\forall x. (C(x) \rightarrow \exists y.(r(x,y) \land X(y)) \land Y(x))]$.

As discussed in \cref{sec:introduction}, second order semantics is stronger than schema semantics.
Let us consider some examples for the \ac{DL} $\EL$, which have more second order entailments than schema semantics for which all $\EL$ concepts are included in the concept base.
These examples will help us to determine the restrictions for the use of concept variables, under which both semantics coincide.

First, we give a minor modification of the example with axiom $\top\sqsubseteq\exists r.X$ discussed in \cref{sec:introduction}, which does not result in an inconsistent ontology under the second-order semantics.

\begin{example}\label{ex:not-range-restricted}
	Consider the ontology $\KB = \set{A  \sqsubseteq \exists r.X}$.
    It is easy to see that $\KB\models^2 \exists r.A\sqsubseteq A$.
    Indeed, assume that $\I\models^2\KB$, and take any valuation $\eta$ such that $\eta(X)=\emptyset$.
    Then $A^\I=A^{\I,\eta}\subseteq(\exists r.X)^{\I,\eta}=\emptyset$.
    Hence $(\exists r.A)^\I=\emptyset\subseteq A^\I$.
	At the same time $\KB \not\models^{\ast}_{\EL} \exists r.A\sqsubseteq A$. 
    Indeed, consider the interpretation $\I=(\Delta^\I,\cdot^\I)$ with
		$\Delta^{\I} = \set{a,b}$,
		$A^{\I}=\set{a}$,
        $B^\I=\set{a,b}$ for every $B\in N_C\setminus\set{A}$,
        and $r^\I=\set{\tuple{a,a},\tuple{b,a}}$ for every $r\in N_R$.
    By structural induction, it is easy to show that $a\in C^\I$ for every $\EL$ concept $C$, hence $\I\models A\sqsubseteq\exists r.C$.
    Therefore, $\I\models^{\ast}_{\EL} \KB$.
    However, since $(\exists r.A)^\I=\set{a,b}\not\subseteq\set{a} = A^\I$, we obtain $\KB \not\models^{\ast}_{\EL} \exists r.A
    \sqsubseteq A$.    
\end{example}

Example~\ref{ex:not-range-restricted} can be generalized to many other \gel axioms $C\sqsubseteq D$ for which there exists a concept variable $X$ appearing in $D$ but not in $C$.
In this case, $\I\models^2 C\sqsubseteq D$ implies that $C^{\I,\eta}=\emptyset$ for every valuation $\eta$ because for the extension $\eta'$ of $\eta$ with $\eta'(X)=\emptyset$, we obtain $C^{\I,\eta}=C^{\I,\eta'}\subseteq D^{\I,\eta'}=\emptyset$.

\begin{taggedblock}{arxiv}
\begin{lem}\label{lem:emptyEta}
	Let $\I$ be an
	interpretation, $D$ a \gel concept containing a concept variable $X$ and $\eta$ a valuation such that $\eta(X)=\emptyset$.
	Then $D^{\I, \eta} =\emptyset$.
\end{lem}

\begin{proof}
	The proof is by structural induction over $D$. 
  Note that the cases $D=\top$, $D=A\in N_C$, and $D=Y\in N_X\setminus\set{X}$ are not possible since $D$ must contain $X$.
	\begin{itemize}
		\item If $D = X$ then $D^{\I, \eta} =\eta(X)=\emptyset$.
		\item If $D = \exists r.E$ then $E$ contains $X$. By induction hypothesis, $E^{\I,\eta}=\emptyset$.
                 Hence $D^{\I,\eta}=\emptyset$.
		\item If $D = E_1 \sqcap E_2$ then $E_1$ or $E_2$ contain $X$.
            By induction hypothesis, $E_1^{\I,\eta}=\emptyset$ or $E_2^{\I,\eta}=\emptyset$.            
            Hence $D^{\I,\eta}= E_1^{\I,\eta}\cap E_2^{\I,\eta}=\emptyset$.
            \qedhere
	\end{itemize}
\end{proof}
\end{taggedblock}

Now, take any $\I\models^2 C\sqsubseteq D$. 
Then $C^{\I,\eta}=\emptyset$ for every valuation $\eta$.
Then $\theta(C)^\I=\emptyset$ for every concept variable substitution $\theta$.
Then for every $\EL$ concept $E$ such that $\theta(C)\in\sub(E)$ we have $E^\I=\emptyset$ as well.
Thus $\set{C\sqsubseteq D}\models^2 E\sqsubseteq F$ for every $F$.
If the schema semantics preserves all these entailments then, in particular, all such concepts $E$ (containing instances of $C$) must be equivalent.
This can happen only in some trivial cases, e.g., when the $\KB$ contains axiom of the form $X\sqsubseteq Y$, which implies that all concepts are equivalent.
To ensure that the semantics coincide in non-trivial cases, it is, therefore, make sense to require that all variables that are present on the right side of a concept
inclusion axioms are also present on the left side. 
Axioms that fulfill this requirement we called \emph{range restricted} axioms.

\cref{ex:schema:conservativity} presents another situation when schema semantics gives fewer consequences than the second-order semantics.
As has been shown in this example, for an $\ELX$ axiom $\alpha$ \eqref{ax:EL} and an $\EL$ ontology $\KB$ from \eqref{eq:conservativity:KB}, we have $\KB\cup\set{\alpha}\not\models^{\ast}_{\EL} A\sqsubseteq D$, however, since $\set{\alpha}\models^{\ast}_{\ALC}\beta$ \eqref{ax:ALC} and $\KB\cup\set{\beta}\models^{\ast}_{\EL}A\sqsubseteq D$, we have $\KB\cup\set{\alpha}\models^2 A\sqsubseteq D$.

The problem with $\alpha$ in this example is that the variable $X$ occurs twice on the left side of the axiom, which makes this axiom equivalent to $\beta$ under the second-order semantics, and, consequently, being able to express an axiom with a concept disjunction $\exists r.C \sqcap \exists s.D \sqsubseteq \exists t.(C \sqcup D)$, which otherwise could not be expressed by ordinary $\EL$ axioms, i.e., under the schema semantics.
To prevent such a case for our fragment of \gel, it therefore, makes sense to prohibit the occurrence of the same variable twice on the left side of an axiom.
Concepts in which each variable occurs at most once, we call \emph{linear}.

To motivate our last restriction, consider the next example.

\begin{example}
    \label{ex:not-safe}
    Consider the \gel ontology $\KB = \set{\exists r. X  \sqsubseteq \exists s.(X \sqcap A)}$.
    It is easy to see that $\KB \models^2 \exists r.\top \sqsubseteq \exists r.A$.
    Indeed, take any $\I\models^2\KB$ and $d\in (\exists r.\top)^\I$.
    Then there exists $d'\in\Delta^\I$ such that $\tuple{d,d'}\in r^\I$.
    Take any valuation $\eta$ with $\eta(X)=\set{d'}$.
    Since $\tuple{d,d'}\in r^\I$, we have $d\in (\exists r.X)^{\I,\eta}$.
    Since $\I\models^2_\eta\KB$, we have $d\in (\exists s.(X\sqcap A))^{\I,\eta}$.
    In particular, $\emptyset\neq (X\sqcap A)^{\I,\eta}=\set{d'}\cap A^\I$.
    Hence $d'\in A^\I$. 
    Consequently, $d\in (\exists r.A)^\I$.
    
    On the other hand, $\KB \not\models^{\ast}_{\EL} \exists r.\top \sqsubseteq \exists r.A$, as evidenced by the counter-model $\I=(\Delta^\I,\cdot^\I)$ with
    $\Delta^{\I}  = \set{a,b}$,
    $A^{\I}  = \set{a}$,
    $r^{\I}  = \set{\tuple{a, b}}$,
    $s^{\I}  = \set{\tuple{a, a}}$,
    and $E^\I=\emptyset$, $h^\I=\emptyset$ for any remaining $E\in N_C$ and $h\in N_R$.
    To show that $\I\models^{\ast}_\EL\KB$, we prove that $\I\models \exists r.C\sqsubseteq\exists s.(C\sqsubseteq A)$ for every $\EL$ concept $C$.
    For this, take any $d\in (\exists r.C)^\I$.
    By definition of $r^{\I}$, $d=a$ and $b\in C^\I$.
    Then $C$ can be only a conjunction of concept $\top$.
    Hence $a\in C$.
    Hence $d=a\in(\exists s.C)^\I$.
    Thus $\I\models\KB$.
    Since, $a\in(\exists r.\top)^\I$ but $(\exists r.A)^\I=\emptyset$, we proved that $\KB \not\models^{\ast}_{\EL} \exists r.\top \sqsubseteq \exists r.A$.
\end{example}

Note that under the second order semantics, the axiom $\exists r.X\sqsubseteq \exists s.(X\sqcap A)$ in $\KB$ from \cref{ex:not-safe} implies two properties: (1) that $r$ is a \emph{subrole} of $s$ ($r \sqsubseteq s$), which is equivalent to axiom $\exists r.X \sqsubseteq \exists s.X$, and (2) that $A$ is a \emph{range} of the role $r$ ($ran(r) \sqsubseteq A$), which is due to the fact that for any element $d'$ such that $\tuple{d,d'}\in r^\I$ this axiom holds for $X=\set{d'}$.
As was shown in the example, the schema semantics cannot capture the second kind of properties.
In fact, an extension of $\EL$ with both (complex) role inclusions and range restrictions becomes undecidable \cite{baader2008pushing}), which could explain why the schema semantics cannot characterize consequences in this extension.
To prevent situations like in \cref{ex:not-safe}, we require that variables in the right side of axioms appear only directly under existential restrictions.
We generalize a related notion of \emph{safe nominals} \cite{kazakov2012practical} to define this restriction:

\begin{defn}[Safe Concept]\label{def:safeConcepts}\index{safe!concept}
	A \gel concept $C$ is called \emph{safe} (for concept variables), if variables only occur in the
	form of $\exists r.X$, i.e.\ safe concepts are defined by the grammar:
	\[
		C_s^{(i)} = A \mid \top \mid \exists r.X \mid \exists r. C_s \mid C^{1}_s \sqcap C^{2}_s
	\]
\end{defn}


\section{When Semantics Coincide}
\label{sec:coincide}

In this section we prove that the restrictions on the use of concept variables discussed in \cref{sec:semantics} are sufficient to guarantee that the logical consequences under the schema semantics and second-order semantics coincide.
Towards this goal, we define a fragment \gelo of $\ELX$ that satisfies these restrictions:

\begin{defn}[\gelo]\label{def:sensibleGel}
	A \gel axiom $\beta = C \sqsubseteq D$ is in the fragment \gelo, if
	\begin{itemize}
        \item $\beta$ is range restricted, i.e.\ $\vars(D)\subseteq\vars(C)$,
        \item $C$ is linear, i.e.\ $C$ it does not contain a variable twice, and
		\item $D$ is safe (cf. \cref{def:safeConcepts}).
	\end{itemize}
\end{defn}

An interesting consequence of the first two restrictions in \cref{def:sensibleGel}, is the so-called \emph{singleton property} for valuations.
Intuitively, for checking entailment over the second-order semantics, it is sufficient to consider only valuations that assign concept variables to singleton subsets of the domain.

\begin{defn}[Singleton Valuation]
    A valuation $\eta$ is a \emph{singleton
        valuation}\index{singleton!valuation} if
        $\forall X: \sizeof{\eta(X)} = 1$.
\end{defn}

\begin{lem}\label{lem:shrinkEta}
    Let $\I$ be an interpretation, $C$ a linear \gel concept, $\eta$ a valuation, and $a \in C^{\I, \eta}$. 
    Then there is a singleton valuation $\eta'$ such that $\eta'(X)\subseteq \eta(X)$ for every $X\in\vars(C)$ and $a \in C^{\I, \eta'}$.
\end{lem}

\begin{taggedblock}{arxiv}
\begin{proof}
		Proof by induction over the structure of $C$:
		\begin{itemize}
			\item If $\vars(C)=\emptyset$ then every singleton valuation $\eta'$ works since $a\in C^{\I, \eta'}=C^{\I,\eta}$.
			\item If $C = X$ then $a \in X^{\I, \eta} = \eta(X)$ let
			      $\eta'(X) = \{a\}$ then $\sizeof{\eta'(X)} = 1$, $\eta'(X) \subseteq
				      \eta(X)$ and $a \in \eta'(X) = X^{\I, \eta'}$.
			\item If $C = \exists r.E$ then $a \in (\exists r.E)^{\I, \eta}$ and there exists $b \in \Delta^{\I}$ such that $\tuple{a,b} \in r^{\I}$ and $b \in E^{\I,\eta}$.
                    By induction hypothesis, there exists a singleton $\eta'$ such that $\eta'(X)\subseteq \eta(X)$ for every $X \in \vars(E)$ such that $b \in E^{\I, \eta'}$.
                    Then $\vars(C)=\vars(E)$ and $a \in (\exists r.E)^{\I, \eta'}$.
			\item If $C = E_1 \sqcap E_2$ then $a \in (E_1 \sqcap E_2)^{\I,
				      \eta}=E_1^{\I, \eta}\cap E_2^{\I, \eta}$. 
                 By induction hypothesis for $a\in E_1^{\I, \eta}$, there exists a singleton $\eta_1'$ such that $\eta_1'(X) \subseteq \eta(X)$ for $X\in\vars(E_1)$ and $a\in E_1^{\I, \eta_1'}$.
                 Likewise, for $a\in E_2^{\I, \eta}$, there exists a singleton $\eta_2'$ such that $\eta_2'(X) \subseteq \eta(X)$ for $X\in\vars(E_2)$ and $a\in E_2^{\I, \eta_2'}$.
                 Define $\eta'(X):=\eta_1'(X)$ for $X\in\vars(E_1)$ and $\eta'(X):=\eta_2'(X)$ for $X\notin\vars(E_1)$.
                 Obviously, $\eta'$ is singleton, $\eta'(X)\subseteq\eta(X)$ for $X\in\vars(C)=\vars(E_1)\cup\vars(E_2)$, and $E_1^{\I,\eta'}=E_1^{\I,\eta_1'}$.
                 Since $C$ is linear, $\vars(E_1)\cap\vars(E_2)=\emptyset$, hence $\eta'(X)=\eta_2'(X)$ for $X\in\vars(E_2)$.
                 Hence $E_2^{\I,\eta'}=E_2^{\I,\eta_2'}$.
                 Finally, since $a\in E_1^{\I, \eta_1'}=E_1^{\I, \eta'}$ and $a\in E_2^{\I, \eta_2'}=E_2^{\I, \eta'}$, we obtain $a\in C^{\I,\eta'}$ as required.
                 \qedhere
		\end{itemize}
\end{proof}
\end{taggedblock}

\begin{lem}
    \label{lem:growEta}
    Let $\I$ be an
	interpretation, $D$ an \gel concept, and $\eta'$, $\eta$ valuations such that $\eta'(X)\subseteq\eta(X)$ for every $X\in\vars(D)$.
    Then $D^{\I, \eta'}\subseteq D^{\I, \eta}$.
\end{lem}

\begin{taggedblock}{arxiv}
\begin{proof}
		Proof by induction over the structure of $D$:
		\begin{itemize}
			\item If $\vars(D)=\emptyset$, then, trivially, $D^{\I,\eta'}=D^{\I,\eta}$.
			\item If $D = X$ then $X^{\I, \eta'} = \eta'(X)\subseteq\eta(X)=X^{\I, \eta}$.
			\item If $D = \exists r.E$ then $E^{\I, \eta'}\subseteq E^{\I,\eta}$ by the induction hypothesis.
          Then for any $a\in (\exists r.E)^{\I, \eta'}$ there exists $b\in E^{\I, \eta'}\subseteq E^{\I,\eta}$ such that $\tuple{a,b}\in r^\I$.
          Hence $a\in (\exists r.E)^{\I, \eta}$.
			\item If $C = E_1 \sqcap E_2$ then $E_1^{\I, \eta'}\subseteq E_1^{\I,\eta}$ and $E_2^{\I, \eta'}\subseteq E_2^{\I,\eta}$ by the induction hypothesis.
            Hence $(E_1 \sqcap E_2)^{\I,\eta'}=E_1^{\I, \eta'}\cap E_2^{\I, \eta'}\subseteq E_1^{\I, \eta}\cap E_2^{\I, \eta}=(E_1 \sqcap E_2)^{\I,\eta}$.
            \qedhere
		\end{itemize}
\end{proof}
\end{taggedblock}

By combining the above two lemmata, we obtain the required singleton property:

\begin{thm}[Singleton Property]\label{thm:singletonModel}
	Let $\I$ be an interpretation, $\beta = C \sqsubseteq D$ a range restricted \gel axiom such that $C$ is linear, and $\I\models^2_{\eta'}\beta$ for every singleton valuation $\eta'$.
    Then $\I\models^2\beta$.
\end{thm}

\begin{taggedblock}{arxiv}
\begin{proof}
    Take any valuation $\eta$.
    We need to prove that $\I\models^2_\eta \beta$, i.e., $C^{\I,\eta}\subseteq D^{\I,\eta}$.
	Take any $a \in C^{\I, \eta}$. 
    Then, by \cref{lem:emptyEta}, there is a singleton $\eta'$ with $\eta'(X)\subseteq \eta(X)$ for every $X \in \vars(C)$ such that $a \in C^{\I, \eta'}$.
    Then $\I\models^2_{\eta'}\beta$ by the condition of the theorem, so $a \in C^{\I, \eta'}\subseteq D^{\I, \eta'}$.
	Since $\beta$ is range restricted, we have $\eta'(X)\subseteq \eta(X)$ for every $X \in \vars(D)\subseteq\vars(C)$.
    Then $a\in D^{\I, \eta'}\subseteq D^{\I, \eta}$ by \cref{lem:growEta}.    
	Since $a$ was chosen arbitrarily, this proves that $\I \models^{2}_{\eta} \beta$.
\end{proof}
\end{taggedblock}

Our next goal is to prove that if $\KB\models^2\alpha$ for some \gelo ontology $\KB$ and a (ground) $\EL$ axiom $\alpha$, then $\KB\models^{\ast}_H\alpha$ for some concept base $H$ consisting of (ground) $\EL$ concepts, i.e., $\KB_{\downarrow H}\models\alpha$ (cf. \cref{def:semantics:schema}).
Since $\KB_{\downarrow H}$ is an ordinary $\EL$ ontology, we can use the well known characterization of $\EL$ entailment using so-called $\EL$ \emph{canonical models} (also called \emph{completion graphs} \cite{baader2005pushing}).
Usually canonical models are defined using consequences of the ontology derived by certain inference rules (see, e.g., \cite{kazakov2014incredible}), however the following simplified definition is sufficient for our purpose.
To avoid confusion with $\ELX$ ontologies, we denote $\EL$ ontologies by $\KBG$.

\begin{defn}[Canonical Interpretation]
    \label{def:canonicalModel}
    Let $H$ be a nonempty \emph{concept base} and $\KBG$ an $\EL$ ontology.
    The \emph{canonical interpretation} (w.r.t. $\KBG$ and $H$) is an interpretation $\I=\I(\KBG,H)=(\Delta^\I,\cdot^\I)$ defined by:
		$\Delta^\I  = \set{x_C \mid C \in H}$,
		$A^\I = \set{x_C \mid \KBG \models C \sqsubseteq A}$ for $A\in N_C$, and
		$r^\I = \set{\tuple{x_C, x_D}\in\Delta^\I\times\Delta^\I \mid \KBG \models C \sqsubseteq \exists r.D}$ for $r\in N_R$.	
\end{defn}

If we include all $\EL$ concepts into the concept base $H$, it can be shown that the canonical interpretation $\I$ is a model of $\KBG$ that satisfies exactly $\EL$ axioms entailed by $\KBG$: $\I\models\alpha$ iff $\KBG\models\alpha$.
However, to ensure the latter property for a \emph{fixed} $\alpha$, it is sufficient to use a smaller concept base $H$ that contains only certain sub-concepts of $\KBG$ and $\alpha$.
It turns out that if $H$ satisfies certain conditions for $\KBG=\KB_{\downarrow H}$ and $\alpha$ then $\KB\models^2\alpha$ implies $\I\models\alpha$, which implies $\KBG\models\alpha$.
To prove that $\I\models\alpha$, we show that $\I\models^2\KB$, for which we use the singleton valuation property from  \cref{thm:singletonModel}.
There is a nice connection between singleton valuations in $\I$ and substitutions of concept variables with concepts from $H$, which we are going to exploit.

\begin{defn}[Canonical Substitution]
    \label{def:canonical:substitution}
    We say that a concept variable substitution $\theta$ is \emph{canonical}\index{substitution!canonical} for a concept base $H$, if $\theta(X)\in H$ for every $X\in N_X$.
    Given a canonical interpretation $\I$ (for $\KBG$ and $H$) and a singleton valuation $\eta$ for $\I$, the \emph{canonical substitution determined} by $\eta$ is the substitution $\theta=\theta_{\eta}$ defined by $\theta(X) = C$ iff $\eta(X) = \set{x_C}\subseteq\Delta^\I$.
\end{defn}

The connection to canonical substitutions in \cref{def:canonical:substitution} helps us to characterize interpretations of (complex) $\ELX$ concepts under canonical models and singleton valuations.

\begin{lem}
	\label{lem:singletonInSaturation}
    Let $\I$ be the canonical interpretation w.r.t. some $\EL$ ontology $\KBG$ and a concept base $H$, $F$ an $\ELX$ concept, $\eta$ a singleton valuation, and $\theta = \theta_{\eta}$ the corresponding canonical substitution. 
    Then: 
    \begin{equation}
     F^{\I,\eta}\subseteq\set{x_D\in\Delta^\I\mid \KBG\models D\sqsubseteq\theta(F)}.   
    \end{equation}    
    In addition, if $F$ is a safe concpet (see Definition~\ref{def:safeConcepts}) and 
    $\theta(D)\in H$ for every $\exists r.D\in\sub(F)$ then
    \begin{equation}
     F^{\I,\eta}\supseteq\set{x_D\in\Delta^\I\mid \KBG\models D\sqsubseteq\theta(F)}.   
    \end{equation}    
\end{lem}

\begin{taggedblock}{arxiv}
\begin{proof}
	Proof by induction over the definition of (safe) $\EL$ concept $F$:
	\begin{itemize}
		\item $F = \top$. Then $\theta(F)=\top$ and $\top^{\I,\eta}=\top^\I=\Delta^\I=\set{x_D\in\Delta^\I\mid\KBG\models D\sqsubseteq\top}$.
		\item $F = A\in N_C$. Then $\theta(F)=A$ and $A^{\I,\eta}=A^\I=\set{x_D\in\Delta^\I\mid\KBG\models D\sqsubseteq A}$ by definition of $A^\I$.  
		\item $F = X\in N_X$ and $\eta(X)=\set{x_C}$. Then $\theta(F)=C$ and $F^{\I,\eta}=\eta(X)=\set{x_C}\subseteq\set{x_D\in\Delta^\I\mid \KBG\models D\sqsubseteq C}$.
        Since $F=X$ is not safe, the converse direction does not apply.
        \item $F = \exists r.X$ and $\eta(X)=\set{x_C}$. Then $\theta(F)=\exists r.C$ and $F^{\I,\eta}=\set{x_D\mid \tuple{x_D,x_C}\in r^\I}= \set{x_D\mid\KBG\models D\sqsubseteq \exists r.C}$ by definition of $r^\I$.        
		\item $F = \exists s.G$. Then $\theta(F)=\exists s.\theta(G)$ and $F^{\I,\eta}=\set{x_D\mid \tuple{x_D,x_C}\in s^\I\:\&\: x_C\in G^{\I,\eta}}\subseteq
          \set{x_D\mid \KBG\models D\sqsubseteq\exists s.C\:\&\: \KBG\models C\sqsubseteq \theta(G)}\subseteq \set{x_D\mid \KBG\models D\sqsubseteq\exists s.\theta(G)}$ by definition of $s^\I$, induction hypothesis for $G$, and basic entailment properties.
  
         If $F$ is safe and $\theta(D)\in H$ for every $\exists r.D\in\sub(\theta(F))$ then either $G=X$ (the previous case) or $G$ is safe, in which case, by induction hypothesis,
         $G^{\I,\eta}\supseteq\set{x_C\mid \KBG\models C\sqsubseteq \theta(G)}$.
         Since $\theta(G)\in H$ and $\KBG\models\theta(G)\sqsubseteq\theta(G)$, this implies that $x_{\theta(G)}\in G^{\I,\eta}$. 
         Then $F^{\I,\eta}=\set{x_D\mid \tuple{x_D,x_C}\in s^\I\:\&\: x_C\in G^{\I,\eta}}\supseteq \set{x_D\mid \tuple{x_D,x_{\theta(G)}}\in s^\I}=\set{x_D\mid \KBG\models D\sqsubseteq\exists s.\theta(G)}=\set{x_D\mid \KBG\models D\sqsubseteq \theta(F)}$.
          
		\item $F = F_1 \sqcap F_2$. Then $\theta(F)=\theta(F_1)\sqcap\theta(F_2)$ and $F^{\I,\eta}=F_1^{\I,\eta}\cap F_2^{\I,\eta}\subseteq\set{x_D\in\Delta^\I\mid \KBG\models 
        D\sqsubseteq\theta(F_1)}\cap\set{x_D\in\Delta^\I\mid \KBG\models D\sqsubseteq\theta(F_2)}=\set{x_D\in\Delta^\I\mid \KBG\models D\sqsubseteq\theta(F_1)\sqcap\theta(F_2)}$
        by induction hypothesis and basic entailment properties.

        If $F$ is safe and $\theta(D)\in H$ for every $\exists r.D\in\sub(\theta(F))$ then so are $F_1$ and $F_2$, thus $F^{\I,\eta}=F_1^{\I,\eta}\cap F_2^{\I,\eta}\supseteq \set{x_D\in\Delta^\I\mid \KBG\models D\sqsubseteq\theta(F_1)}\cap\set{x_D\in\Delta^\I\mid \KBG\models D\sqsubseteq\theta(F_2)}=\set{x_D\in\Delta^\I\mid \KBG\models D\sqsubseteq\theta(F_1)\sqcap\theta(F_2)}$
        by induction hypothesis.  
         \qedhere
	\end{itemize}
\end{proof}
\end{taggedblock}

The restriction to safe concepts in \cref{lem:singletonInSaturation} is necessary, as the following example shows:

\begin{example}
	Consider $H=\set{A,B}\subseteq N_C$ and $\KBG = \set{A \sqsubseteq B}$.
    Then the canonical interpretation $\I$ for $\KBG$ and $H$ has domain $\Delta^\I=\set{x_A,x_B}$ and assigns $A^\I=\set{x_A,x_B}$, $B^\I=\set{x_B}$.
    Now take $F=X\in N_X$ and define a singleton $\eta(X)=\set{x_B}$.
    Then the corresponding canonical substitution $\theta(X)=B$.
    As can be seen, $F^{\I,\eta}=\eta(X)=\set{x_B}\subsetneq\set{x_A,x_B}=\set{x_D\in\Delta^\I\mid \KBG\models D\sqsubseteq B}$.
\end{example}

We are now ready to harvest the fruits of our characterization of canonical interpretations.
We first show that, the canonical interpretation cannot entail more axioms than the ontology $\KBG$ for which it is constructed if the concept base $H$ contains relevant concepts from these axioms.

\begin{coro}
    \label{cor:canonical:entail}
    Let $\I$ be the canonical interpretation w.r.t. $\KBG$ and $H$,
    and $\alpha=F\sqsubseteq G$ an $\EL$ axiom such that $F\in H$ and $D\in H$ for every $\exists r.D\in\sub(F)$.
    Then $\I\models\alpha$ implies $\KBG\models\alpha$.
\end{coro}

\begin{proof}    
    By \cref{lem:singletonInSaturation},
    $G^\I\subseteq\set{x_D\in\Delta^\I\mid \KBG\models D\sqsubseteq G}$ and
    $F^\I\supseteq\set{x_D\in\Delta^\I\mid \KBG\models D\sqsubseteq F}$
    Since $\KBG\models F\sqsubseteq F$ and $F\in H$, we have $x_F\in F^\I$.
    Since $\I\models F\sqsubseteq G$, we have $x_F\in G^\I\subseteq\set{x_D\in\Delta^\I\mid \KBG\models D\sqsubseteq G}$.
    Hence $\KBG\models F\sqsubseteq G$.    
\end{proof}

Next, we determine what to put into $H$ and $\KBG$ so that the canonical interpretation $\I=\I(G,H)$ satisfies a given $\ELX$ axiom $\beta$ (to eventually ensure that $\I\models\KB$).

\begin{coro}
    \label{cor:canonical:model}
    Let $\I$ be the canonical interpretation w.r.t. $\KBG$ and $H$, 
    and $\beta=F\sqsubseteq G$ an \gelo axiom such that 1) $\beta_{\downarrow H}\in\KBG$ and 2) $D_{\downarrow H}\subseteq H$ for every $\exists r.D\in\sub(G)$.
    Then $\I\models^2 \beta$.
\end{coro}

\begin{taggedblock}{arxiv}
\begin{proof}
    Since $\beta$ is an \gelo axiom, $\beta$ is range restricted and $F$ is linear.
    Therefore, by \cref{thm:singletonModel}, $\I\models^2 \beta$ iff $F^{\I,\eta}\subseteq G^{\I,\eta}$ for every singleton valuation $\eta$.
    Now, let $\eta$ be any such singleton valuation and $\theta$ be the corresponding canonical substitution.
    Then by \cref{lem:singletonInSaturation}, $F^{\I,\eta}\subseteq\set{x_D\in\Delta^\I\mid \KBG\models D\sqsubseteq\theta(F)}$.
    Since $\beta$ is an \gelo axiom, $G$ is a safe concept.
    Furthermore, $\theta(D)\in D_{\downarrow H}\subseteq H$ for every $\exists r.D\in\sub(G)$ by Condition 2).
    Therefore, by \cref{lem:singletonInSaturation}, $G^{\I,\eta}\supseteq\set{x_D\in\Delta^\I\mid \KBG\models D\sqsubseteq\theta(G)}$.
    Finally, by Condition 1) $\theta(\beta)=\theta(F)\sqsubseteq\theta(G)\in\KBG$, therefore
    $F^{\I,\eta}\subseteq\set{x_D\in\Delta^\I\mid \KBG\models D\sqsubseteq\theta(F)}\subseteq \set{x_D\in\Delta^\I\mid\KBG\models D\sqsubseteq\theta(G)}\subseteq G^{\I,\eta}$.
\end{proof}
\end{taggedblock}

By combining \cref{cor:canonical:entail} and \cref{cor:canonical:model}, for the given \gelo ontology $\KB$ and an $\EL$ axiom $\alpha$, we can define the \emph{smallest} concept base $H$ such that the canonical interpretation $\I$ w.r.t. $\KBG=\KB_{\downarrow H}$ and $H$ satisfies all axioms $\beta\in\KB$ under the second-order semantics and entails $\alpha$ only if $\KBG\models\alpha$ and thus only if $\KB\models^{\ast}\alpha$.
Note that Condition  2) in \cref{cor:canonical:model} is recursive over $H$.
Therefore, the required $H$ is defined as a fixed point limit for this condition.

\begin{defn}[Expansion \& Expansion Base]
    \label{def:expansionBase}
    Let $\KB$ be a $\ELX$ ontology and $\alpha = F \sqsubseteq E$ an $\EL$ axiom.
    Let $H^0=\set{F}\cup\set{D\mid\exists r.D \in\sub(F)}$, and $H^{i+1}=H^i\cup\bigcup_{\exists r.D\in \sub^+(\KB)} D_{\downarrow H^i}$ for $i\ge 0$.
    We call $H^{\infty} =\bigcup_{i \geq 0} H^{i}$ the \emph{expansion base} and $\KBG=\KB_{\downarrow H^{\infty}}$ the \emph{expansion} for $\KB$ w.r.t $\alpha$.
\end{defn}

We show that the expansion base $H^\infty$ is indeed a fixed point of the required condition:

\begin{lem}
    \label{lem:Fixpoint}
    Let $H^\infty$ be the expansion base for $\KB$ w.r.t. $\alpha$.
    Then $D_{\downarrow H^{\infty}}\subseteq H^{\infty}$ for every $\exists r.D\in\sub^+(\KB)$.
\end{lem}

\begin{taggedblock}{arxiv}
\begin{proof}
Take any $D$ with $\exists r.D\in\sub^+(\KB)$ and $D'\in D_{\downarrow H^\infty}$.
Then $D'=D[X_1/C_1,\dots,X_n/C_n]$ for $\set{X_1,\dots,X_n}=\vars(D)$ and $\set{C_1,\dots,C_n}\subseteq H^\infty=\bigcup_{i \geq 0} H^{i}$.
Since $H^i\subseteq H^{i+1}$ for every $i\ge 0$, then $\set{C_1,\dots,C_n}\subseteq H^i$ for some $i\ge 0$.
Then $D'\in D_{\downarrow H^i}\subseteq H^{i+1}\subseteq H^\infty$.
\end{proof}
\end{taggedblock}

By combining \cref{cor:canonical:entail}, \cref{cor:canonical:model} and \cref{lem:Fixpoint}, we now prove the following result:

\begin{thm}
	\label{thm:canonicalModelK}
    Let $\KB$ be an \gelo ontology, $\alpha$ an $\EL$ axiom, and $\KB^\infty$ the expansion of $\KB$ w.r.t. $\alpha$.
    Then $\KB\models^2\alpha$ implies $\KB^\infty\models\alpha$.
\end{thm}

\begin{taggedblock}{arxiv}
\begin{proof}    
    Let $H^\infty$ be the expansion base for $\KB$ and $\alpha$, and $\I$ the canonical interpretation for $\KB^\infty$ and $H^\infty$.
    By \cref{lem:Fixpoint}, conditions 1) and 2) of \cref{cor:canonical:model} are satisfied for $H=H^\infty$, $\KBG=\KB^\infty$ and every $\beta\in\KB$.
    This implies $\I\models^2\KB$.
    Then $\KB\models^2\alpha$ implies $\I\models\alpha$.    
    Furthermore, by \cref{def:expansionBase}, the conditions of \cref{cor:canonical:entail} are satisfied for $H=H^\infty \supseteq H^{0}$, $\KBG=\KB^\infty$ and $\alpha$.
    Hence, $\I\models\alpha$ implies $\KBG=\KB^\infty\models\alpha$, as required.    
\end{proof}
\end{taggedblock}

An immediate consequence of \cref{thm:canonicalModelK} is that the schema semantics coincides with second-order semantics for \gelo ontologies.

\begin{thm}\label{thm:SemanticsCoincide}
	Let $\KB$ be a \gelo ontology and $\alpha$ an $\EL$ axiom.
	Then $\KB \models^{2} \alpha \Leftrightarrow \KB \models^{\ast}_\EL\alpha$.
\end{thm}

\begin{taggedblock}{arxiv}
\begin{proof}
		$(\Rightarrow)$:
            By \cref{thm:canonicalModelK}, $\KB \models^{2} \alpha$ implies $\KB^\infty\models\alpha$. Then
            $\KB\models^{\ast}_{H^\infty}\alpha$, and so, $\KB\models^{\ast}_{\EL}\alpha$.
            
		$(\Leftarrow)$:
            By \cref{def:semantics:schema}, $\KB \models^{\ast}_\EL\alpha$ means that
            $\KB_{\downarrow H}\models\alpha$ where $H$ is the set of all $\EL$ concepts.
            In order to prove that $\KB \models^2\alpha$, take any $\I\models^2\KB$.
            We prove that $\I\models\alpha$.
            We first show that $\I\models\KB_{\downarrow H}$.
            Indeed, take any $\beta'\in\KB_{\downarrow H}$.
            Then $\beta'=\theta(\beta)$ for some substitution $\theta$ of concept variables to concepts from $H$.
            Define a valuation $\eta$ by $\eta(X)=\theta(X)^\I$.
            By induction on the structure of $\EL$ concepts $D$ it is easy to show that $D^{\I,\eta}=\theta(D)^\I$.
            Then $\I\models\beta'$ follows from $\I\models^2_\eta\beta\in\KB$ since $\I\models^2\KB$.
            This proves $\I\models\KB_{\downarrow H}$.
            Since $\KB_{\downarrow H}\models\alpha$, we obtain $\I\models\alpha$ as required.
\end{proof}
\end{taggedblock}


\section{Decidability}
\label{sec:decidabilty}

Because the schema semantics and the second-order semantics coincide for \gelo, we immediately obtain semi-decidability of the entailment for the latter.
In general, entailment in \gelo ist still \emph{undecidable} because \gelo can express (unrestricted)  role-value-maps:

\begin{defn}[Role-Value-Maps]\index{role-value-maps}\label{def:roleValueMap}
    A \emph{role-value-map} \cite{baader2003restricted} is an axiom of the form
		$r_1 \circ \dots \circ r_m \sqsubseteq s_1 \circ \dots \circ s_n$
	with $m,n \geq 1$, $r_i,s_j\in N_R$ ($1\le i\le m$, $1\le j\le n$).
    The \emph{interpretation of role-value-maps} is defined by:
		$\I \models r_1 \circ \dots \circ r_m \sqsubseteq s_1 \circ \dots \circ s_n \text{ iff } r_1^{\I} \circ \dots \circ r_m^{\I}
		\subseteq s_1^{\I} \circ \dots \circ s_n^{\I}$, where $\circ$ is the usual composition of binary relations.
\end{defn}

\begin{lem}\label{lem:roleValueMapEmulation}
	For every interpretation $\I$ it holds
		$\I \models^{2} \exists r_1.\exists r_2. \dots \exists r_m.X \sqsubseteq \exists s_1.\exists s_2. \dots \exists s_n.X$
	iff
		$r_1^{\I} \circ \dots \circ r_m^{\I}
		\subseteq s_1^{\I} \circ \dots \circ s_n^{\I}$.
\end{lem}

\begin{taggedblock}{arxiv}
\begin{proof}
	\
	\begin{itemize}
		\item[$\Rightarrow$:]
			Assume that
			\[
				\I \models^{2} \exists r_1.\exists r_2. \dots \exists r_m.X \sqsubseteq \exists s_1.\exists s_2. \dots \exists s_n.X
			\]
			Then we need to show that
			\begin{align*}
				\phantom{\forall x,z \in
				\Delta^{\I}:\ } &
				\begin{aligned}
					\mathllap{\forall x,z \in
						\Delta^{\I}:\ }
					(\exists y_1, \dots, y_{m-1} \in
					\Delta^{\I}: & \langle x,y_1
					\rangle \in r_1^{\I}, \langle y_1, y_2 \rangle \in
					r_2^{\I}, \dots,                              \\
					                      & \langle y_{m-1}, z \rangle \in
					r_m^{\I})
				\end{aligned}
				\\
				\phantom{\Rightarrow\ }  &
				\begin{aligned}
					\mathllap{\Rightarrow\ } (\exists y_1, \dots, y_{n-1} \in
					\Delta^{\I}: & \langle x,y_1
					\rangle \in s_1^{\I}, \langle y_1, y_2 \rangle \in
					s_2^{\I}, \dots,                                                 \\
					                      & \langle y_{n-1}, z \rangle \in s_n^{\I})
				\end{aligned}
			\end{align*}
			Take any $x, z \in \Delta^{\I}$ such that
			\[
				\exists y_1, \dots, y_{m-1} \in
				\Delta^{\I}: \langle x,y_1
				\rangle \in r_1^{\I}, \langle y_1, y_2 \rangle \in
				r_2^{\I}, \dots, \langle y_{m-1}, z \rangle \in r_m^{\I}
			\]
			Let $\eta$ be a valuation such that $\eta(X) = \{z\}$. Then $x \in
				(\exists r_1.\exists r_2. \dots \exists r_m.X)^{\I, \eta}$
			and it follows that $x \in (\exists s_1.\exists s_2. \dots \exists
				s_n.X)^{\I, \eta}$. Then
			\[
				\exists y_1, \dots, y_{n-1} \in
				\Delta^{\I}: \langle x,y_1
				\rangle \in s_1^{\I}, \langle y_1, y_2 \rangle \in
				s_2^{\I}, \dots, \langle y_{n-1}, z \rangle \in s_n^{\I}
			\]
		\item[$\Leftarrow$:]
			Assume that
			\begin{align*}
				\phantom{\forall x,z \in
				\Delta^{\I}:\ } &
				\begin{aligned}
					\mathllap{\forall x,z \in
						\Delta^{\I}:\ }
					(\exists y_1, \dots, y_{m-1} \in
					\Delta^{\I}: & \langle x,y_1
					\rangle \in r_1^{\I}, \langle y_1, y_2 \rangle \in
					r_2^{\I}, \dots,                              \\
					                      & \langle y_{m-1}, z \rangle \in
					r_m^{\I})
				\end{aligned}
				\\
				\phantom{\Rightarrow\ }  &
				\begin{aligned}
					\mathllap{\Rightarrow\ } (\exists y_1, \dots, y_{n-1} \in
					\Delta^{\I}: & \langle x,y_1
					\rangle \in s_1^{\I}, \langle y_1, y_2 \rangle \in
					s_2^{\I}, \dots,                                                 \\
					                      & \langle y_{n-1}, z \rangle \in s_n^{\I})
				\end{aligned}
			\end{align*}
			We need to show
			that
			\[
				\I \models^{2} \exists r_1.\exists r_2. \dots \exists r_m.X \sqsubseteq \exists s_1.\exists s_2. \dots \exists s_n.X
			\]
			Take any $x$ and $\eta$ then if $x \in (\exists r_1.\exists r_2.
				\dots \exists r_m.X)^{\I, \eta}$ then
			\begin{align*}
				\exists y_1, \dots, y_{m-1}, z \in
				\Delta^{\I}: & \langle x,y_1
				\rangle \in r_1^{\I}, \langle y_1, y_2 \rangle \in
				r_2^{\I}, \dots,                              \\
				                      & \langle y_{m-1}, z \rangle \in
				r_m^{\I}, z \in \eta(X)
			\end{align*}
			Then from the role-value-map axiom it follows that
			\[
				\exists y_1, \dots, y_{n-1} \in
				\Delta^{\I}: \langle x,y_1
				\rangle \in s_1^{\I}, \langle y_1, y_2 \rangle \in
				s_2^{\I}, \dots, \langle y_{n-1}, z \rangle \in s_n^{\I}
			\]
			Then $x \in (\exists s_1.\exists s_2. \dots \exists s_n.X)^{\I, \eta}$. As $x$ and $\eta$ were chosen
			arbitrarily, this proves the proposition.
            \qedhere
	\end{itemize}
\end{proof}
\end{taggedblock}

\begin{thm}
	Axiom entailment in \gelo is undecidable.
\end{thm}

\begin{proof}
	Follows directly from \cref{lem:roleValueMapEmulation} and the fact that
	axiom entailment is undecidable in $\EL$ extended with role-value-maps \cite{baader2003restricted}.
\end{proof}

The reason for this undecidiblity, is deep nesting of concept variables on the right side of axioms under existential restrictions.
Such nested variables result in infinite expansion base $H^{\infty}$ as the following example shows:

\begin{example}
    Consider the \gelo ontology $\KB = \set{X \sqsubseteq \exists r. \exists r.X}$ and $\alpha=A \sqsubseteq B$.
    Then according to \cref{def:expansionBase}, we have $
    H^1 = \set{A}$, 
    $\KB^1 = \set{A \sqsubseteq \exists r. \exists r.A}$,
    $H^2 = \set{A, \exists r.A}$, 
    $\KB^2 = \set{A \sqsubseteq \exists r. \exists r.A, \exists r.A \sqsubseteq \exists r. \exists r.\exists r.A}$,
    $H^3 = \set{A, \exists r.A, \exists r.\exists r.A}$, etc.
\end{example}

If we restrict \gelo so that variables on the right side do not appear under nested existential restrictions, we can show that the expansion $\KB^{\infty}$ of the ontology is, in fact, polynomial in the size of $\KB$, which gives us polynomial decidability of the (schema and second-order) entailment.

\begin{defn}[\gelt]
    \label{def:restriction:decidability}
	An \gelt axiom is an \gelo axiom $\alpha$ such that for
    every $\exists r.D\in\sub^+(\alpha)$, either $D=X\in N_X$ or $\vars(D)=\emptyset$.
\end{defn}

\begin{lem}
	\label{lem:domainsEqual}
	Let $\KB$ be an \gelt ontology, $\alpha$ an $\EL$ axiom and $H^{\infty}$ the expansion base for $\KB$ and $\alpha$.
    Then $H^{\infty}=H^0\cup \set{D\mid \exists r.D\in\sub^+(\KB)\:\&\:\vars(D)=\emptyset}$.
\end{lem}

\begin{taggedblock}{arxiv}
\begin{proof}
    By \cref{def:expansionBase}, $H^{i+1}=H^i\cup\bigcup_{\exists r.D\in\sub^+(\KB)} D_{\downarrow H^i}$ for $i\ge 0$.
    By \cref{def:restriction:decidability}, $\exists r.D\in\sub^+(\KB)$ implies $D=X\in N_X$ or $\vars(D)=\emptyset$.
    If $D=X$ then $D_{\downarrow H^i}= X_{\downarrow H^i}=H^i$.
    If $\vars(D)=\emptyset$ then $D_{\downarrow H^i} = \set{D}$.
    Then $H^{i+1}= H^i\cup \set{D\mid \exists r.D\in\sub^+(\KB)\:\&\:\vars(D)=\emptyset}$.    
    Hence $H^\infty=\bigcup_{i\ge 0} H^i \subseteq H^0\cup\set{D\mid \exists r.D\in\sub^+(\KB)\:\&\:\vars(D)=\emptyset}$.    
\end{proof}
\end{taggedblock}

Since the elements of $H^\infty$ are subsets of (ground) concepts appearing in $\KB$ and $\alpha$, we obtain:

\begin{thm} 
    \label{thm:polynomial}
	Let $\KB$ be a \gelt ontology and $\alpha$ an $\EL$ axiom. 
    Then the entailment $\KB \models^{2} \alpha$ is decidable in polynomial time in the size of $\KB$ and $\alpha$.
\end{thm}

\begin{taggedblock}{arxiv}
\begin{proof}
    By \cref{lem:domainsEqual}, $H^{\infty}=H^0\cup \set{D\mid \exists r.D\in\sub^+(\KB)\:\&\:\vars(D)=\emptyset}\subseteq\sub^-(\alpha)\cup\sub^+(\KB)$.
    Hence $\sizeof{H^\infty}\le \sizeof{\KB}+\sizeof{\alpha}$.
    Since $\KB^\infty=\KB_{\downarrow H^\infty}$, $\sizeof{\KB^\infty}\le\sizeof{\KB}\cdot\sizeof{H^\infty}^V$ where $V$ is the maximal number of concept variables in axioms of $\KB$.
    It can be shown that, w.l.o.g., $V\le 1$ for \gelt ontologies.    
    Hence, $\sizeof{\KB^\infty}\le \sizeof{\KB}\cdot(\sizeof{\KB}+\sizeof{\alpha})$.
    
    By \cref{thm:canonicalModelK}, $\KB \models^{2} \alpha$ implies $\KB^\infty \models \alpha$.
    Conversely, $\KB^\infty \models \alpha$ implies $\KB\models^{\ast}_\EL\alpha$, which implies $\KB \models^{2} \alpha$ by \cref{thm:SemanticsCoincide}.
    Hence, entailment $\KB \models^{2} \alpha$ can be reduced to the classical $\EL$ entailment in polynomial time.
    Since the latter is polynomially decidable \cite{baader2005pushing}, this implies that the entailment $\KB \models^{2} \alpha$ is also decidable in polynomial time.    
\end{proof}
\end{taggedblock}

\section{Conclusions and Outlook}
\label{sec:conclusions}

In this paper, we left behind the usual restriction of \acp{DL} to fragments of
\ac{FOL} and introduced concept variables directly into \acp{DL}.
These concept variables can be understood as simple placeholders for concrete
concepts, giving us axiom schemas. This results in a semantics that replaces
variables by a specific set of concepts and in this way reduces reasoning to the classical
case. Or they can be understood more strongly, as universally quantified
concepts, similar to predicates in \ac{SOL}. This gives us a semantics that
interprets variables as arbitrary subsets of the interpretation domain and
results in the DL being a fragment of \ac{SOL}.

We applied this extension to $\EL$ and analyzed the difference in
entailed conclusions by the two semantics. We defined a fragment for which the
conclusions coincide, given us semi-decidability also for second-order semantics.
We also showed that for a slight limitation of this fragment, second-order semantics (and
schema semantics) even become decidable.

In this decidable fragment \gelt, we can express a range of features that
usually require special constructors in classical $\EL$: We can express role chain axioms that reduce a chain of roles to a
	      connection via one role. What would normally be expressed as $r_1 \circ
		      \dots \circ r_n \sqsubseteq s$ (i.e.\ role-value-maps where the right
	      side is a single role, cf. \cref{def:roleValueMap}), we can express as
	      $\exists r_1. \exists r_2. \dots \exists r_n.X \sqsubseteq \exists s.X$.
	      For example $\textsf{father} \circ \textsf{father} \sqsubseteq
		      \textsf{grandfather}$ is equivalent to $\exists
		      \textsf{father}.\exists \textsf{father}.X \sqsubseteq \exists
		      \textsf{grandfather}.X$. We can also express self restrictions on the right side of axioms,
	      i.e.\ classical axioms of the form $C \sqsubseteq \exists
		      r.\mathrm{Self}$, meaning $\forall x: (x \in
		      C^{\I}) \Rightarrow (\langle x, x \rangle \in
		      r^{\I})$ (cf.\ e.g.\ \cite{wang2013consequence}). We express this as $C \sqcap X \sqsubseteq \exists r.X$. For example
	      $\textsf{GreatApes} \sqsubseteq \exists \textsf{recognize}.\mathrm{Self}$ is
	      equivalent to $\textsf{GreatApes} \sqcap X \sqsubseteq \exists
		      \textsf{recognize}.X$.
         We can express positive occurrences of (local) \emph{role-value-map concepts}, i.e., concepts of the form $r \subseteq s$ interpreted as $\set{x\mid \forall y: \tuple{x,y}\in r^{\I} \Rightarrow \tuple{x,y}
    \in s^{\I}}$ (cf.\ e.g.\ \cite{donini2003complexity}). 
    E.g., we express $C\sqsubseteq (r \subseteq s)$ as $C \sqcap \exists r.X \sqsubseteq
    \exists s.X$. 
    For example $\textsf{Male} \sqsubseteq (\textsf{isParentOf}
    \subseteq \textsf{isFatherOf})$ is equivalent to $\textsf{Male} \sqcap \exists \textsf{isParentOf}. X \sqsubseteq
    \exists \textsf{isFatherOf}. X$.
    Finally, we can express restrictions that generalize \emph{all} of these constructs, for example, axioms of the form:
    \begin{equation}
    \label{ax:generalized}
    C_0\sqcap \exists r_1.(C_1\sqcap\exists r_2.(C_2\dots\sqcap\exists r_n.(C_n\sqcap X)\dots))\sqsubseteq \exists s.X,\quad (n\ge 0)
    \end{equation}
Note that \eqref{ax:generalized} can be expressed with self-restrictions over fresh roles: $C_i\sqsubseteq\exists h_i.\text{Self}$ $(0\le i\le n)$ and role chain axiom:
$h_0\circ r_1\circ h_2\circ r_2\cdots r_n\sqsubseteq s$.
Thus, it is not clear whether our decidable fragment has more expressive power than known polynomial extensions of $\EL$ \cite{baader2008pushing, wang2013consequence}. 

It is possible to generalize \cref{thm:SemanticsCoincide} and \cref{thm:polynomial} to entailments of arbitrary $\ELX$ axioms $\alpha$ because such entailments can be always reduced to entailment of $\EL$ axioms by simply replacing all concept variables in $\alpha$ with new atomic concepts, not occurring in $\KB$
or $\alpha$. 
Intuitively, to prove that every (second order or schema) model of the ontology satisfies $\alpha$ under each valuation $\eta$ or substitution $\theta$, we can extend this model by interpreting the new atomic concepts according to the value of $\eta(X)$ or $\theta(X)$ on the variables $X$ which they replace. 
This still remains a model of the ontology since it does not contain these new concepts.

Summarizing, the results in this paper show that an extension of \acp{DL} to
being fragments of \ac{SOL}, instead of \ac{FOL}, is possible (while remaining
decidable for reasonable restrictions) and does allow expressing facts in a new
way without the use of special constructors. This enlightens the relationship
between \ac{FOL}, \ac{SOL} and \acp{DL} and opens the way for further extensions
of \acp{DL} outside \ac{FOL}.

\bibliography{paper}

\appendix

\section{Acronyms}

\begin{acronym}
	\acro{ss}[Schema-S]{schema semantics}
	\acro{sos}[SO-S]{second-order semantics}
    \acro{DL}[DL]{Description Logic}
    \acro{PL}{Propositional Logic}
    \acro{FOL}{First-Order Logic}
    \acro{SOL}{Second-Order Logic}
    \acro{ODP}{Ontology Design Pattern}
    \acro{ML}{Modal Logic}
\end{acronym}
\end{document}